\documentclass{amsart}

\usepackage{amssymb,amsmath,amscd,amsthm,enumerate}
\usepackage[utf8]{inputenc}

\date{\today}

\newcommand{\Hi}{{\mathcal H}}
\newcommand{\Z}{{\mathbb Z}}
\newcommand{\R}{{\mathbb R}}
\newcommand{\C}{{\mathbb C}}

\newcommand{\D}{{\mathbb D}}

\newcommand{\supp}{{\mathrm{supp}}}

\newcommand{\CMV}{{\mathcal C}}

\newtheorem{theorem}{Theorem}[section]
\newtheorem{lemma}[theorem]{Lemma}
\newtheorem{prop}[theorem]{Proposition}
\newtheorem{coro}[theorem]{Corollary}

\theoremstyle{definition}
\newtheorem{remark}[theorem]{Remark}
\theoremstyle{definition}

\theoremstyle{definition}

\allowdisplaybreaks

\numberwithin{equation}{section}

\renewcommand{\Re}{\mathrm{Re} \, }

\newcommand{\abs}[1]{\left\vert#1\right\vert}
\newcommand{\set}[1]{\left\{#1\right\}}
\newcommand{\eqdef}{\overset{\mathrm{def}}=}
\newcommand{\norm}[1]{\left\|#1\right\|}

\begin{document}

\title[Quantum Intermittency for Sparse CMV Matrices]{Quantum Intermittency for Sparse CMV Matrices with an Application to Quantum Walks on the Half-Line}

\author[D.\ Damanik]{David Damanik}

\address{Department of Mathematics, Rice University, Houston, TX~77005, USA}

\email{damanik@rice.edu}

\author[J.\ Erickson]{Jon Erickson}

\address{Department of Mathematics, Rice University, Houston, TX~77005, USA}

\email{jfe23@rice.edu}

\author[J.\ Fillman]{Jake Fillman}

\address{Department of Mathematics, Rice University, Houston, TX~77005, USA}

\email{fillman@vt.edu}

\author[G.\ Hinkle]{Gerhardt Hinkle}

\address{Department of Mathematics, Rice University, Houston, TX~77005, USA}

\email{gnh1@rice.edu}

\author[A.\ Vu]{Alan Vu}

\address{Department of Mathematics, Rice University, Houston, TX~77005, USA}

\email{atv3@rice.edu}

\begin{abstract}

We study the dynamics given by the iteration of a (half-line) CMV matrix with sparse, high barriers. Using an approach of Tcheremchantsev, we are able to explicitly compute the transport exponents for this model in terms of the given parameters. In light of the connection between CMV matrices and quantum walks on the half-line due to Cantero-Gr\"unbaum-Moral-Vel\'azquez, our result also allows us to compute transport exponents corresponding to a quantum walk which is sparsely populated with strong reflectors. To the best of our knowledge, this provides the first rigorous example of a quantum walk that exhibits quantum intermittency, i.e., nonconstancy of the transport exponents. When combined with the CMV version of the Jitomirskaya-Last theory of subordinacy and the general discrete-time dynamical bounds from Damanik-Fillman-Vance, we are able to exactly compute the Hausdorff dimension of the associated spectral measure.

\end{abstract}

\maketitle


\section{Introduction}

Recently, quantum walks have been studied heavily; see \cite{AVWW, BGVW, CGMV, CGMV2,DFO, DFV, DMY2, J11, J12, JM, K14, KS11, KS14, ST12} (and references therein) for some papers on this subject that have appeared in the past five years. These are quantum analogues of classical random walks. In this paper, we will concentrate on the special case of quantum walks on the half-line; let us briefly describe the appropriate setting. We consider a system whose state may be described by a vector $(v_n)_{n=0}^\infty$ with $v_0 \in \C$, $v_n \in \C^2$ for $n \geq 1$ which is $\ell^2$-normalized in the sense that
$$
|v_0|^2 + \sum_{n=1}^\infty \| v_n \|^2
=
1,
$$
so that $|v_0|^2$ may be thought of as the probability that the state is at the origin and $\|v_n\|^2$ corresponds to the probability that the state is at site $n \in \Z_+$.   A quantum walk on $\Z_+$ is then described by choosing ``quantum coins,'' which are unitary $2 \times 2$ matrices, and using these to determine the probability that a given state will transition to the left or the right after one time unit. We shall describe this in more detail in Subsection~\ref{sec:qw} below. The upshot, following an important observation of Cantero, Gr\"unbaum, Moral, and Vel\'azquez in \cite{CGMV}, is this: the evolution of the quantum walk during one time unit may now be described by a CMV matrix $\CMV$ acting on $\ell^2(\Z_0)$ (note that we write $\Z_+$ for the set of positive integers and $\Z_0$ for the set of nonnegative integers). This draws an interesting analogy between the present setting and the case of classical random walks on $\Z_+$ with symmetric nearest-neighbor interactions, which may be parameterized by Jacobi matrices in a natural fashion. In fact, the analogy goes deeper than that, since CMV matrices are the natural unitary analogues of Jacobi matrices, and, moreover, they play a canonical role within the class of unitary operators analogous to the canonical role played by Jacobi matrices within the class of (bounded) self-adjoint operators; see \cite{S1, S2} and references therein. Thus, we are concerned with the study of
\begin{equation}\label{e.CMVevolution}
\langle \delta_n , \CMV^t v \rangle,
\quad
t\in \Z, \; n \in \Z_0,
\end{equation}
where $\CMV$ is a CMV matrix, and in particular a unitary operator on $\ell^2(\Z_0)$. This allows one to employ spectral theoretical methods to analyze the behavior of the corresponding quantum walk, since expressions of the form \eqref{e.CMVevolution} may be rewritten as suitable integrals against spectral measures of $\CMV$. Indeed, the inner product in \eqref{e.CMVevolution} is a Fourier coefficient of a spectral measure, so one can prove quantitative decay estimates for it in terms of the fractal regularity of the associated spectral measure by using the theory which proves decay of the Fourier coefficients of a measure in terms of the regularity of the same \cite{DFV}.

In addition to this approach, there is an alternative method that we employ here, and which was proposed in \cite{DFV} and developed in \cite{DFO}; namely, one may rewrite \eqref{e.CMVevolution} by means of integrals of matrix elements of the resolvent of $\CMV$ over the unit circle with respect to normalized Lebesgue/Haar measure on $\partial \D$ (i.e., the measure generated by normalized arc length). This connects time-averaged spreading of wave packets to properties of transfer matrices, since transfer matrices are related to the properties of the resolvent probed by the integral in the fundamental formula. This connection is exceedingly useful in many cases of interest as there are many ways to study the growth of transfer matrix norms. In the present paper, we combine these two approaches \`a la \cite{T2005}, and we are able to rigorously observe the phenomenon of \emph{quantum intermittency} in a quantum walk that is sparsely populated with strong reflectors. We accomplish this by proving a slightly more general result which computes transport exponents for a class of CMV matrices with sparse, high barriers, and then applying the CGMV connection. To the best of our knowledge, this represents the first class of examples of quantum walks with explicitly computed nontrivial (i.e. $\neq 0, 1$) transport exponents.

\bigskip

Let us now describe the models and results more carefully.  For each $\alpha \in \D \eqdef \set{z \in \C : \abs{z} < 1}$, define the unitary matrix $\Theta(\alpha)$ by
$$
\Theta(\alpha)
\eqdef
\begin{pmatrix}
\overline{\alpha} & \rho \\
\rho & - \alpha
\end{pmatrix},
\quad
\rho
\eqdef
\sqrt{1-|\alpha|^2}.
$$
Given a sequence $(\alpha_n) \in \D^{\Z_0}$, the corresponding \emph{CMV matrix} is a linear operator on $\Hi \eqdef \ell^2(\Z_0)$ given by $\CMV = \mathcal L \mathcal M$, where
$$
\mathcal L
=
\begin{pmatrix}
\Theta(\alpha_0) &&& \\
& \Theta(\alpha_2) && \\
&& \Theta(\alpha_4) & \\
&&& \ddots
\end{pmatrix}
,
\quad
\mathcal M
=
\begin{pmatrix}
1 &&& \\
& \Theta(\alpha_1) && \\
&& \Theta(\alpha_3) & \\
&&& \ddots
\end{pmatrix}.
$$
It is straightforward to check that $\CMV$ enjoys the matrix representation
\begin{equation} \label{def:cmv}
\CMV
=
\begin{pmatrix}
 \overline{\alpha_0} & \overline{\alpha_1}\rho_0 & \rho_1\rho_0 &&& & \\
\rho_0 & -\overline{\alpha_1}\alpha_0 & -\rho_1 \alpha_0 &&& & \\
& \overline{\alpha_2}\rho_1 & -\overline{\alpha_2}\alpha_1 & \overline{\alpha_3} \rho_2 & \rho_3\rho_2 & & \\
& \rho_2\rho_1 & -\rho_2\alpha_1 & -\overline{\alpha_3}\alpha_2 & -\rho_3\alpha_2 &  &  \\
&&& \overline{\alpha_4} \rho_3 & -\overline{\alpha_4}\alpha_3 & \overline{\alpha_5}\rho_4 & \rho_5\rho_4 \\
&&& \rho_4\rho_3 & -\rho_4\alpha_3 & -\overline{\alpha_5}\alpha_4 & -\rho_5 \alpha_4  \\
&&&& \ddots & \ddots &  \ddots & \ddots
\end{pmatrix},
\end{equation}
where all unspecified matrix entries of $\CMV$ are zero. One typically refers to $(\alpha_n)_{n = 0}^\infty$ as the sequence of \emph{Verblunsky coefficients} of $\CMV$. There is a close relationship between operators of the form \eqref{def:cmv} and the theory of orthogonal polynomials on the unit circle \cite{S1,S2}. Let us briefly note that $\delta_0$ is a cyclic vector for $\CMV$. Whenever we refer to ``the spectral measure'' of $\CMV$, we mean the spectral measure associated to $\CMV$ and $\delta_0$, that is, the unique Borel measure on $\partial \D$ that satisfies
\[
\langle \delta_0, g(\CMV)\delta_0 \rangle
=
\int_{\partial \D} \! g \, d\mu
\]
for all continuous functions $g:\partial \D \to \C$.

CMV matrices furnish a particularly pleasant class of \emph{quantum walks}, i.e., dynamical systems defined by the iteration of a unitary operator on the unit sphere of a Hilbert space. There is a (very) rough physical heuristic which identifies the size of $\rho_n^{-1}$ with the size of the barrier through which a wavepacket must tunnel to escape the region $[0,n]$. To lend a small air of credibility to one direction of this heuristic, one may consider $\mathcal C_0$, the CMV matrix obtained by setting $\alpha_n \equiv 0$ for all $n$. It is straightforward to verify that $\mathcal C_0^k \delta_0 = \delta_{2k-1}$ for all $k \in \Z_+$, so the wavepacket propagates ballistically (i.e., without impediment from the environment). This heuristic is refined considerably and elucidated further in the main theorems and examples of \cite{DFO}.

\bigskip

Let us now precisely describe what we mean by CMV matrices with sparse, high barriers. First, we fix a sequence of integers $1 \leq L_1 < L_2 < \cdots $ with the property that
\begin{equation} \label{eq:sparsedef}
\nu_N
\eqdef
\frac{\log(L_1 L_2 \cdots L_{N-1})}{\log L_N}
\to 0
\end{equation}
as $N \to \infty$. Now, fix $\eta \in (0,1)$, and let us define a sequence of Verblunsky coefficients by
\begin{equation} \label{eq:sparseCMVdef}
\alpha_{n}
\eqdef
\begin{cases}
\sqrt{1-L_j^{-\frac{1-\eta}{\eta}}} & \text{ if } n = L_j \text{ for some } j, \\
0 & \text{ otherwise.}
\end{cases}
\end{equation}
According to the heuristic above, the wavepacket encounters a barrier of size $L_k^{\frac{1-\eta}{2\eta}}$ at each site $L_k$ and encounters no other impediments. Given the condition \eqref{eq:sparsedef}, the size of the barriers and the separation between consecutive obstacles becomes very large; hence the term ``sparse, high barriers.'' We are interested in the time evolution of the initial state $\delta_0 \in \ell^2(\Z_0)$, that is, we want to study $\CMV^t \delta_0$ as $t \in \Z_+$ grows. To quantify the dynamics, we first put
$$
a(n,t)
\eqdef
\left| \left\langle \delta_n , \CMV^t \delta_0 \right\rangle \right|^2,
\quad
t \in \Z, \, n \in \Z_0,
$$
which can be thought of as the probability that the associated wave packet is at the site $n$ at time $t$. We shall also be interested in the time-averaged probabilities, given by
$$
\widetilde{a}(n,T)
\eqdef
\left( 1-e^{-2/T} \right) \sum_{t = 0}^{\infty} e^{-2t / T} a(n,t),
\quad
T > 0, \, n \in \Z_0.
$$
We will also frequently consider the inside and outside probabilities, given by
$$
P(n < M, T)
\eqdef
\sum_{0 \leq n < M} \widetilde{a}(n,T),
\qquad
P(n \geq M,T)
\eqdef
\sum_{n \geq M} \widetilde{a}(n,T)
$$
for $M , T > 0$.

The following formula, from \cite[Lemma~3.16]{DFV}, allows us to connect time averages of dynamical quantities to averages of matrix elements of the resolvent of $\CMV$. We have
$$
\sum_{t = 0}^\infty e^{-2t/T} \left| \langle \varphi, \CMV^t \delta_0 \rangle \right|^2
=
e^{2/T} \int_0^{2\pi} \! \left| \left\langle \varphi, (\CMV - e^{i\theta+1/T})^{-1} \delta_0 \right\rangle \right|^2 \, \frac{d\theta}{2\pi}
$$
for any $T > 0$ and any $\varphi \in \Hi$. In particular,
\begin{equation} \label{eq:parsevalprob}
\widetilde{a}(n,T)
=
(e^{2/T} - 1) \int_0^{2\pi} \!\left| \left\langle \delta_n, (\CMV - e^{i\theta+1/T})^{-1} \delta_0 \right\rangle \right|^2 \, \frac{d\theta}{2\pi}
\end{equation}
for all $n \in \Z_0$, $T > 0$. Since $\CMV$ is a unitary operator, $\| \CMV^t \delta_0 \| = 1$ for every $t \in \Z$, so we may then think of $\CMV^t \delta_0$ as defining a probability distribution on $\Z_0$. Consequently, one may describe the spreading of these distributions in terms of their moments. More precisely, for $p > 0$ and $T > 0$, define
\[
\left\langle |X|^p \right\rangle(T)
\eqdef
\sum_{n = 0}^\infty ( n^p +1 ) \widetilde{a}(n,T).
\]
We would like to compare the growth of the $p$th moment to polynomial growth of the form $T^{\beta p}$ for a suitable exponent $\beta \in [0,1]$. Thus, the following transport exponents\footnote{Some authors consider Ces\`aro averages for the moments, instead of the exponential averages which we consider. However, it is not hard to see that either method of averaging yields the same values for $\widetilde\beta^\pm$; compare \cite[Lemma~2.19]{DT2010}} are natural objects to consider
\[
\widetilde{\beta}^+(p)
\eqdef
\limsup_{T \to \infty} \frac{\log \left( \left\langle |X|^p \right\rangle (T) \right)}{p \log(T)},
\quad
\widetilde{\beta}^-(p)
\eqdef
\liminf_{T \to \infty} \frac{\log \left( \left\langle|X|^p\right\rangle(T) \right)}{p \log(T)}.
\]
By Jensen's inequality $\widetilde{\beta}^+$ and $\widetilde{\beta}^-$ are both non-decreasing functions of $p$; compare \cite[Lemma~2.7]{DT2010}.

Our main result is an exact CMV analog of the main result of \cite{T2005} for discrete half-line Schr\"odinger operators with growing sparse potentials; one can precisely compute the transport exponents, and the lower transport exponent is a strictly increasing function of $p$. This is a phenomenon known as \emph{quantum intermittency} in the physics literature \cite{Mantica}. Physically (and somewhat loosely) speaking, Theorem~\ref{t:betas} tells us that the quantum dynamical transport properties of a sparse CMV matrix are \emph{inhomogeneous} in at least two senses. First, the result implies that $\widetilde \beta^-(p) < \widetilde \beta^+(p)$ for all $p > 0$, which implies that wavepackets propagate wildly differently on different time scales. Second, nonconstancy of $\widetilde \beta^-(p)$ as a function of $p$ means that the fastest and slowest parts of the wavepacket travel at different rates of speed. One may find a more substantial discussion of the relationship between $\widetilde \beta^\pm(p)$ and quantum dynamics in \cite{DT2010}.

\begin{theorem} \label{t:betas}
With $\alpha$ and $\CMV$ as above, we have
$$
\widetilde\beta^-(p)
=
\frac{p+1}{p+1/\eta},
\qquad
\widetilde\beta^+(p)
=
1
\text{ for all } p > 0.
$$
\end{theorem}

\medskip

Since CMV matrices are unitary analogs of Jacobi matrices, the broad strokes of the proof of Theorem~\ref{t:betas} proceed along the same general lines as \cite{T2005}, so, aside from some minor deviations, the large-scale structure of the present paper is similar to \cite{T2005}. However, in many cases, the technical nuances of the proofs are fairly distinct. This comes about for several reasons; perhaps most significant is the fact that the dynamical analysis revolves around the resolvent, via the Parseval formula \eqref{eq:parsevalprob}. In the Jacobi matrix case, the primary integral transform of the spectral measure is the Borel transform, which connects directly to the resolvent -- in fact, the Borel transform of the spectral measure of a Jacobi matrix is simply the (1,1) matrix entry of the resolvent! However, the primary integral transform used to probe CMV matrices is the Carath\'eodory function, whose connection to the resolvent is more oblique.\footnote{For a more thorough discussion of Carath\'eodory vs.\ Borel in the OPUC setting, see the excellent paper \cite{SimonJACM2004}.}

\medskip

We also point out that our result allows us to compute the fractal dimension of the spectral measure of $\CMV$ exactly. The following statement follows from the the quantitative version of subordinacy theory due to Jitomirskaya-Last. More specifically, it follows from the arguments which prove \cite[Theorem~1.3(ii)]{JL99}. See Appendix~\ref{sec:speccont} for more details.

\begin{theorem} \label{t:subord}
Let $\alpha$ and $\CMV$ be as above, and let $\mu = \mu_{\delta_0}$ denote the corresponding spectral measure. The lower Hausdorff dimension of $\mu$ is bounded below by $\eta$ in the sense that $\mu$ does not give weight to sets of Hausdorff dimension less than $\eta$.
\end{theorem}

Thus, combining this with our result and the general dynamical bounds in \cite{DFV}, we have the following companion to Theorem~\ref{t:subord}, which allows us to exactly compute the fractal dimension of $\mu$.

\begin{theorem} \label{t:specdim}
With $\alpha$ and $\CMV$ as above, $\mu$ has a support of Hausdorff dimension $\eta$. In particular, $\mu$ has exact {\rm(}Hausdorff{\rm)} dimension $\eta$.
\end{theorem}

\begin{remark}
In fact, our method can be used to prove a statement about the dimension of the spectral measure that is uniform in the boundary condition. Concretely, our dynamical estimates also apply to Verblunsky coefficients of the form $\alpha^\lambda_n = \lambda \alpha_n$ with $\lambda \in \partial \D$ and $\alpha$ given by \eqref{eq:sparsedef}--\eqref{eq:sparseCMVdef}. If one denotes the spectral measure associated to $\alpha^\lambda$ by $\mu^\lambda$, our methods demonstrate that $\mu^\lambda$ has exact dimension $\eta$ for all $\lambda \in \partial \D$. This solves a problem posed by Simon in \cite[Section~12.9]{S2} by improving~\cite[Theorem~12.9.4]{S2}, which shows that $\mu^\lambda$ has exact dimension $\eta$ for Lebesgue a.e.\ $\lambda \in \partial \D$. For simplicity of exposition, we will only work with $\lambda = 1$, but the modifications to make our arguments work for general $\lambda \in \partial \D$ are easy.
\end{remark}

\subsection{Quantum Walks on $\Z_+$} \label{sec:qw}

We now precisely describe quantum walks on the half-line and their relationship with CMV matrices, following \cite{CGMV,DFO,DFV}.  A quantum walk on $\Z_+$ is modeled by a unitary operator on the state space $\Hi_+ = \left( \ell^2(\Z_+) \otimes \C^2 \right) \oplus \langle \delta_0 \otimes e_\downarrow \rangle$, which models a space in which a wave packet comes equipped with a ``spin'' at each positive integer site. The elementary tensors of the form $\delta_n \otimes e_\uparrow$, and $\delta_n \otimes e_\downarrow$ with $n \in \Z_+$ together with $\delta_0 \otimes e_\downarrow$ comprise an orthonormal basis of $\Hi_+$. A time-homogeneous quantum walk scenario is given as soon as unitary coins
\begin{equation}\label{e.timehomocoins}
C_{n}
=
\begin{pmatrix}
c^{11}_{n} & c^{12}_{n} \\
c^{21}_{n} & c^{22}_{n}
\end{pmatrix}
\in \mathrm U(2), \quad n \in \Z_+,
\end{equation}
are specified. Additionally, one specifies an appropriate boundary condition at the origin, e.g.,
$$
C_0
=
\begin{pmatrix} 0 & 1 \\ -1 & 0 \end{pmatrix}.
$$
As one passes from time $t$ to time $t+1$, the update rule of the quantum walk is given by
\begin{align}
\delta_{n} \otimes e_\uparrow & \mapsto
  c^{11}_{n} \delta_{n+1} \otimes e_\uparrow
+ c^{21}_{n} \delta_{n-1} \otimes e_\downarrow , \label{e.updaterule1} \\
\delta_n \otimes e_\downarrow  & \mapsto
  c^{12}_{n} \delta_{n+1} \otimes e_\uparrow
+ c^{22}_{n} \delta_{n-1} \otimes e_\downarrow \label{e.updaterule2}.
\end{align}
If we extend this by linearity and continuity to general elements of $\mathcal{H}_+$, this defines a unitary operator $U$ on $\mathcal{H}_+$. Next, order the basis of $\mathcal{H}_+$ by taking $\phi_{2m-1} = \delta_m \otimes e_\uparrow$, $\phi_{2m} = \delta_m \otimes e_\downarrow$ for $m \in \Z_+$, and $\phi_0 = \delta_0 \otimes e_\downarrow$. In this ordered basis, the matrix representation of $U : \mathcal{H}_+ \to \mathcal{H}_+$ is given by
\begin{equation}\label{e.umatrixrep}
U
=
\begin{pmatrix}
 0 & c_1^{21} & c_1^{22} &&&&& \\
 1 & 0 & 0 &&& && \\
& 0 & 0 & c_2^{21} & c_2^{22} & && \\
& c_1^{11} & c_1^{12} & 0 & 0 & && \\
&&& 0 & 0 & c_3^{21} & c_3^{22}& \\
&&& c_2^{11} & c_2^{12} & 0 & 0 & \\
&&&& \ddots & \ddots &  \ddots & \ddots
\end{pmatrix},
\end{equation}
which is easy to check using the update rule \eqref{e.updaterule1}--\eqref{e.updaterule2}; compare \cite[Section~4]{CGMV}.

We can connect quantum walks to CMV matrices using the following observation. If all Verblunsky coefficients with even index vanish, the CMV matrix in \eqref{def:cmv} becomes
\begin{equation}\label{e.ecmvoddzero}
\small
\CMV
=
\begin{pmatrix}
0 & \overline{\alpha_1}  & \rho_1 &&& & \\
1 & 0 & 0 &&& & \\
& 0 & 0 & \overline{\alpha_3}  & \rho_3 & & \\
& \rho_1 & - \alpha_1 & 0 & 0 &  &  \\
&&& 0 & 0 & \overline{\alpha_5} & \rho_5 \\
&&& \rho_3 & - \alpha_3 & 0 & 0   \\
&&&& \ddots & \ddots &  \ddots
\end{pmatrix}.
\end{equation}
The matrix in \eqref{e.ecmvoddzero} strongly resembles the matrix representation of $U$ in \eqref{e.umatrixrep}. Note, however, that $\rho_n > 0$ for all $n$, so \eqref{e.umatrixrep} and \eqref{e.ecmvoddzero} may not match exactly when $c_n^{kk}$ is not real and positive. However, this can be easily resolved by conjugation with a suitable diagonal unitary. Since this is not an issue in the present paper, we skip the details of the gauge transformation and refer the reader to \cite{CGMV,DFO} for lucid expositions. In particular, the following result follows immediately from Theorem~\ref{t:betas}.

\begin{theorem} \label{t:sparseqw}
Let $1 \leq L_1 < L_2 < \cdots$ be a sequence of integers satisfying \eqref{eq:sparsedef}, and define coins $C_n \in \mathrm{SO}(2)$ by
\[
C_{L_j}
\eqdef
\begin{pmatrix}
r_j & - \sqrt{1-r_j^2} \\
\sqrt{1-r_j^2} & r_j
\end{pmatrix},
\quad
r_j
\eqdef
(2L_j-1)^{-\frac{1-\eta}{2\eta}}
\]
for each $j \in \Z_+$, and $C_n = I$ for all other $n$. Then, with the initial state $\psi = \delta_0 \otimes e_\downarrow$, we have
\[
\widetilde \beta^-(p)
=
\frac{p+1}{p+1/\eta},
\quad
\widetilde \beta^+(p)
=
1
\]
for all $p > 0$.
\end{theorem}

Physically speaking, the entries of the $n$th coin may be thought of as reflection and transmission coefficients at the site $n \in \Z_+$. Concretely, if $C_n = I_2$, the $2 \times 2$ identity matrix, then this corresponds to placing a perfect transmitter at site $n$; on the other hand, if
\[
C_n
=
J
\eqdef
\begin{pmatrix}
0 & -1 \\
1 &  0
\end{pmatrix},
\]
then this corresponds to placing a perfect reflector at site $n$. Since $C_{L_j} \sim J$ for large $j$, $C_{L_j}$ may be thought of as a strong reflector.

Strictly speaking, we have not defined the moments or the transport exponents for unitary operators on $\Hi_+$; however, the definitions are completely similar, and it is easy to check that the unitary equivalence that identifies the update rule of a quantum walk on $\Hi_+$ with a CMV matrix on $\Hi$ preserves the transport exponents.

\bigskip

The structure of the remainder of the paper is as follows. Section~\ref{sec:bg} contains background on the relevant objects needed to tackle the proof of Theorem~\ref{t:betas}. Section~3 proves lower bounds on the moments and transport exponents, and Section~4 proves upper bounds on the same. Finally, we describe the proof of Theorem~\ref{t:subord} in Appendix~\ref{sec:speccont}.

\section*{Acknowledgements}  D.\ D.\ and J.\ F.\ were supported in part by NSF grants DMS--1067988 and DMS--13616125. J.\ E., G.\ H., and A.\ V.\ were supported by NSF grant DMS--1148609.

\section{Background and Preparatory Work} \label{sec:bg}

\subsection{Return Probabilities and the Poisson Kernel}

We will frequently consider the (time-averaged) probability that the wavepacket returns to its initial state, given by
$$
\left( 1 - e^{-2/T} \right) \sum_{t = 0}^\infty e^{-2t/T} |\langle \delta_0, \CMV^t\delta_0 \rangle|^2.
$$
If $\mu = \mu_{\delta_0}$ denotes the spectral measure of $\CMV$, we note that
\begin{align*}
\sum_{t = 0}^\infty e^{-2t/T} |\langle \delta_0, \CMV^t \delta_0 \rangle|^2
& =
\sum_{t = 0}^\infty e^{-2t/T} \int_{\partial \D} \! \int_{\partial \D} \! z^t w^{-t} \, d\mu(z) \, d\mu(w) \\
& =
\frac{1}{2} \int_{\partial \D} \! \int_{\partial \D} \! \left( 1 + P_{e^{-2/T}} (z/w) \right) \, d\mu(z) \, d\mu(w),
\end{align*}
where $P$ denotes the Poisson kernel, defined by
\begin{equation} \label{def:poissonk}
P_r(\tau)
\eqdef
\sum_{\ell \in \Z} r^{|\ell|} \tau^\ell
=
\frac{1-r^2}{1-2r \mathrm{Re}(\tau) + r^2},
\quad
r \in [0,1), \, |\tau| = 1.
\end{equation}
Notice that the second step in the calculation follows from taking the real part of the second expression. Thus, we define
$$
J(\varepsilon)
\eqdef
\frac{1}{2} (1-e^{-2\varepsilon}) \int_{\partial \D} \! \int_{\partial \D} \! \left( 1 + P_{e^{-2\varepsilon}} (z/w) \right) \, d\mu(z) \, d\mu(w),
\quad
\varepsilon >0,
$$
so that the (time-averaged) return probability at time $T$ is simply $J(1/T)$.

\subsection{The Gesztesy--Zinchenko Cocycle}

Often, to study the spectral theory of a CMV matrix $\CMV$, we must examine solutions to the difference equation $\CMV u = zu$ with $z \in \C \setminus \{0\}$ and $u \in \C^{\Z_0}$. To that end, consider the  matrices
\begin{equation} \label{gz:onestepmats:def}
P(\alpha,z)
=
\frac{1}{\rho}
\begin{pmatrix}
-\alpha & z^{-1} \\
z & - \overline{\alpha}
\end{pmatrix}
,
\quad
Q(\alpha,z)
=
\frac{1}{\rho}
\begin{pmatrix}
-\overline{\alpha} & 1 \\
1 & - \alpha
\end{pmatrix},
\, \alpha \in \D, z \in \C \setminus \{0\},
\end{equation}
where $\rho = \rho_\alpha = \left( 1 - |\alpha|^2\right)^{1/2}$ as before. Notice that
$$
\det(P(\alpha,z))
=
\det(Q(\alpha,z))
=
-1
\text{ for all } \alpha \in \D, \, z \in \C \setminus \{0\}.
$$
These matrices come from \cite{GZ06}, though our $\alpha_n$ is their $-\overline{\alpha_{n+1}}$. One may use $P$ and $Q$ to capture the recursion described by the difference equation $\CMV u = zu$, in a sense which we presently describe. Since these matrices are absolutely central to our work, and our normalization is different from that of \cite{GZ06}, we provide a proof of the following proposition for the convenience of the reader.

\begin{prop} \label{p:gz:stepbystep}
Let $u \in \C^{\Z_0}$ be such that $\CMV u = zu$, and define $v \eqdef \mathcal M u$, $\Phi(n) \eqdef (u(n), v(n))^\top$. For all $n \ge 0$, we have
\begin{equation} \label{eq:gz:stepbystep}
\Phi(n+1)
=
\begin{cases}
Q(\alpha_n,z)
\Phi(n)
& \text{if } n \text{ is odd} \\
P(\alpha_n,z)
\Phi(n)
& \text{if } n \text{ is even}.
\end{cases}
\end{equation}

\end{prop}

\begin{proof}
Since $v = \mathcal M u$, we have $u_0 = v_0$ and
\begin{equation} \label{eq:uvtheta1}
\begin{pmatrix} v_{2k-1} \\ v_{2k} \end{pmatrix}
=
\Theta(\alpha_{2k-1})
\begin{pmatrix} u_{2k-1} \\ u_{2k} \end{pmatrix}
\end{equation}
for each $k \geq 1$. Looking at the first coordinate of \eqref{eq:uvtheta1}, we get
$$
v_{2k-1}
=
\overline{\alpha_{2k-1}} u_{2k-1} + \rho_{2k-1} u_{2k},
$$
and hence
$$
u_{2k}
=
\frac{1}{\rho_{2k-1}} \left(v_{2k-1} - \overline{\alpha_{2k-1}} u_{2k-1} \right).
$$
Similarly, looking at the first coordinate of
$$
\begin{pmatrix} u_{2k-1} \\ u_{2k} \end{pmatrix}
=
\Theta(\alpha_{2k-1})^{-1}
\begin{pmatrix} v_{2k-1} \\ v_{2k} \end{pmatrix}
$$
and using $\Theta(\alpha)^{-1} = \Theta(\overline\alpha)$, we get
$$
v_{2k}
=
\frac{1}{\rho_{2k-1}}\left( u_{2k-1} - \alpha_{2k-1} v_{2k-1}  \right).
$$
These calculations prove \eqref{eq:gz:stepbystep} for odd $n$. Similarly, we note that $\mathcal L v = \CMV u = zu$, so
$$
z\begin{pmatrix} u_{2k} \\ u_{2k+1} \end{pmatrix}
=
\Theta(\alpha_{2k})
\begin{pmatrix} v_{2k} \\ v_{2k+1} \end{pmatrix}
$$
for all $k \geq 0$. Solving this system for $u_{2k+1}$ and $v_{2k+1}$ in terms of $u_{2k}$ and $v_{2k}$ resolves the other case.
\end{proof}

The foregoing proposition motivates the following definition. Denote $Y(n,z) = Q(\alpha_n,z)$ when $n$ is odd and $Y(n,z) = P(\alpha_n,z)$ when $n$ is even; then, the \emph{Gesztesy-Zinchenko cocycle} is defined by
\begin{equation} \label{gz:def}
Z(n,m;z)
=
\begin{cases}
Y(n - 1,z) \cdots Y(m, z) & n > m \\
I & n = m \\
Y(n, z)^{-1} \cdots Y(m - 1, z)^{-1} & n < m
\end{cases}
\end{equation}
If $u$, $v$, and $\Phi$ are as above, we have
\begin{equation} \label{eq:gzsoltransfer}
\Phi(n)
=
Z(n, m; z) \Phi(m)
\text{ for all } n,m \in \Z_+
\end{equation}
by \eqref{eq:gz:stepbystep}; compare \cite[Lemma~2.2]{GZ06}. We will often abbreviate the names and refer to $Z(n,m;z)$ as a GZ matrix.

\subsection{Generalized Eigenfunctions}

Given a CMV matrix $\CMV$, let $\mu = \mu_{\delta_0}$ denote the corresponding spectral measure, and let $W : \Hi \to L^2(\partial \D, \mu)$ denote the canonical unitary equivalence which maps $g(\CMV) \delta_0$ to $g$. The \emph{generalized eigenfunctions} of $\CMV$ are defined by
$$
\xi_n
=
W\delta_n,
\quad
n \ge 0.
$$
These will play an important role in our dynamical analysis. Their signficance in our setting is readily apparent from the identity
$$
\widetilde a(n,T)
=
\frac{1}{2} (1-e^{-2/T}) \int _{\partial \D} \! \int_{\partial \D} \! \left( 1+ P_{e^{-2/T}}(z/w) \right) \xi_n(z) \overline{\xi_n(w)} \, d\mu(z) \, d\mu(w),
$$
which is a straightforward consequence of the definitions and the spectral theorem.

\subsection{The Carath\'eodory Function}
If $\CMV$ is a CMV matrix and $\mu = \mu_{\delta_0}$ is its spectral measure, then the \emph{Carath\'eodory function} of $\CMV$ is defined by
$$
F(z)
=
F_\mu(z)
\eqdef
\int_{\partial \D} \frac{\tau+z}{\tau-z} \, d\mu(\tau).
$$
In the present setting, the Carath\'eodory function often plays a role analogous to that played by the Borel transform for Borel measures on $\R$. One can check that
$$
\mathrm{Re}(F_\mu(z))
=
\int_{\partial\D} \frac{1-|z|^2}{|\tau-z|^2} \, d\mu(\tau)
$$
for all $z \notin \supp(\mu) = \sigma(\CMV)$. Equivalently, writing $z = re^{i\theta}$, one has
\begin{equation} \label{eq:cara:poisson}
\mathrm{Re}(F(z))
=
-\int_{\partial\D} \! P_{1/r}(\tau e^{-i\theta}) \, d\mu(\tau)
\text{ whenever } r > 1.
\end{equation}
We begin by noting the following formula for the integral of the real part of the Carath\'eodory function against arc length measure.

\begin{lemma} \label{l:cara:int}
Let $\CMV$ be any CMV matrix with Carath\'eodory function $F$. For all $\varepsilon > 0$, we have
$$
\int_0^{2\pi} \! \mathrm{Re}\left( F(e^{i\theta +\varepsilon}) \right) \frac{d\theta}{2\pi}
=
-1.
$$
\end{lemma}

\begin{proof}
By \eqref{eq:cara:poisson}, Fubini's Theorem, the definition of the Poisson kernel, and dominated convergence, we have
\begin{align*}
\int_0^{2\pi} \Re \! \left(F(e^{i\theta + \varepsilon})\right) \, \frac{d\theta}{2\pi}
& =
- \int_0^{2\pi} \int_{\partial \D} P_{e^{-\varepsilon}}(\tau e^{-i\theta}) \, d\mu(\tau)\, \frac{d\theta}{2\pi} \\
& =
-\int_{\partial \D} \int_0^{2\pi} \sum_{j \in \Z}  \tau^j e^{-ij\theta-\varepsilon|j|}\,  \frac{d\theta}{2\pi} \, d\mu(\tau) \\
& = -1,
\end{align*}
as desired.
\end{proof}

There are several useful connections between the Carath\'eodory function of $\CMV$ and its resolvent, which we note in the following proposition.

\begin{prop} \label{p:cara:resolve}
Let $\CMV$ be a CMV matrix, $z \notin \sigma(\CMV)$, and $u = (\CMV - z)^{-1} \delta_0$. If $z \neq 0$, we have
\begin{equation} \label{eq:cfct:u0}
u_0
=
\frac{F(z) - 1}{2z}.
\end{equation}
If $z \notin \partial \D$, one also has
\begin{equation} \label{eq:cfct:unorm}
\| u \|^2
=
\frac{\mathrm{Re}(F_\mu(z))}{1-|z|^2}.
\end{equation}
Finally, with $v = \mathcal M u$, we have
\begin{equation} \label{eq:gzcara:resolvent}
\begin{pmatrix} u_n \\ v_n \end{pmatrix}
=
\frac{1}{2z} Z(n,0;z) \begin{pmatrix}
F(z)+1 \\ F(z)-1
\end{pmatrix}
\quad
\text{for all } n \in \Z_+
\end{equation}
whenever $z \notin \sigma(\CMV) \cup \{0\}$.
\end{prop}

\begin{proof}
By definition,
\[
F(z)
=
\int_{\partial \D}\left( 1 + \frac{2z}{\tau-z} \right) \, d\mu(\tau)
=
1 + 2z\langle \delta_0, (\CMV - z)^{-1} \delta_0 \rangle
=
1 + 2zu_0.
\]
Solving for $u_0$ yields \eqref{eq:cfct:u0}.

Next, by linearity of the integral and the spectral theorem,
\[
\mathrm{Re}(F(z))
=
\int_{\partial \D} \! \frac{1-|z|^2}{|\tau - z|^2} \, d\mu(\tau)
=
\left( 1 - |z|^2 \right) \! \| u \|^2,
\]
which proves \eqref{eq:cfct:unorm}.

Let us now turn to \eqref{eq:gzcara:resolvent}. Since $v = \mathcal M u$, we have
\begin{equation} \label{eq:odd:uvtheta}
v_0 = u_0
=
\frac{F(z)-1}{2z},
\text{ and }
\Theta(\alpha_{2k-1})
\begin{pmatrix} u_{2k-1} \\ u_{2k} \end{pmatrix}
=
\begin{pmatrix} v_{2k-1} \\ v_{2k} \end{pmatrix},
\, k \geq 1.
\end{equation}
Similarly, we note that $\mathcal L v = \CMV u = zu + \delta_0$, so
\begin{equation} \label{eq:even:uvtheta}
 \Theta(\alpha_{2k})
\begin{pmatrix} v_{2k} \\ v_{2k+1} \end{pmatrix}
=
z\begin{pmatrix} u_{2k} \\ u_{2k+1} \end{pmatrix}
\text{ for all } k \geq 1.
\end{equation}
Finally, using the first two rows and columns of $\mathcal L v = zu + \delta_0$, we get
$$
\Theta(\alpha_0)
\begin{pmatrix}
v_0 \\
v_1 \\
\end{pmatrix}
=
\begin{pmatrix}
zu_0 + 1 \\
zu_1 \\
\end{pmatrix},
$$
so we get
\begin{equation} \label{eq:zero:uvtheta}
\Theta(\alpha_0) \begin{pmatrix}
v_0 \\
v_1 \\
\end{pmatrix}
=
\begin{pmatrix}
z\widetilde{u}_0 \\
zu_1 \\
\end{pmatrix},
\end{equation}
where $\widetilde{u}_0 = u_0 + \frac{1}{z} = \frac{1}{2z}(F(z)+1)$. Following the proof of Proposition~\ref{p:gz:stepbystep} and using \eqref{eq:odd:uvtheta}, \eqref{eq:even:uvtheta}, and \eqref{eq:zero:uvtheta}, we have
$$
\begin{pmatrix}
u_n \\
v_n \\
\end{pmatrix}
=
Z(n,0;z) \begin{pmatrix}
\widetilde{u}_0 \\
v_0 \\
\end{pmatrix}
$$
for all $n \geq 1$, which proves \eqref{eq:gzcara:resolvent}.
\end{proof}

\section{Lower Bounds}

In this section, we will prove the lower bounds from Theorem~\ref{t:betas}, where we view the two identities as pairs of inequalities, following the general approach of \cite{T2005}. Let us make a brief comment on notation. Throughout the paper, $C$ will denote a constant which depends only on $\eta$ and $\{L_j\}$, and $C_p$ denotes a constant which may additionally depend on $p > 0$; to avoid cluttering the notation, we will use the same letters in all corollaries, lemmas, propositions, and theorems, though they may change from line to line.

We will implement the following overall strategy to prove the desired lower bounds. First, we will use the definition of $(L_j)_{j=1}^\infty$ and $(\alpha_n)_{n=0}^\infty$ to prove upper bounds on the GZ matrices (Lemma~\ref{l:znormbound}). Next, using the connection between the GZ transfer matrices and the resolvent of $\CMV$ from \eqref{eq:gzcara:resolvent} and the Parseval formula, we can translate these into lower bounds on the outside probabilities (Theorem~\ref{t:pbounds}). At last, these lower bounds can then be translated into appropriate lower bounds on the moments (Corollary~\ref{c:momentbound} and Theorem~\ref{t:betalowerbounds}). The precise details follow. We begin with a few preparatory lemmas. First, we relate the return and outside probabilities.

\begin{lemma}\label{l:firstbound}
If $M(T) \eqdef 1 / \left(8J(T^{-1}) \right)$, then
\[
P (n \geq M(T), T)
\geq
\frac{1}{2}
>
0.
\]
\end{lemma}

\begin{proof}
Clearly, it suffices to show
$$
P(n < M(T), T)
\leq
\frac{1}{2}.
$$
Let $\xi_n$ and $\mu = \mu_{\delta_0}$ denote the generalized eigenfunctions and spectral measure of $\CMV$, respectively. Given $M > 0$, $z,w \in \partial \D$, and $T > 0$ define
\begin{align*}
S_M(z, w)
& =
\sum_{n < M} \overline{\xi_n(z)} \xi_n(w), \\
G_M(w)
& =
\int_{\partial \D} |S_M(z,w)|^2 \, d\mu(z), \\
b(w, T)
& =
(1-e^{-2/T})^2 \int_{\partial \D} \left |\frac{1}{1 - e^{-2/T}z/w}\right| ^2 d\mu(z).
\end{align*}
By the spectral theorem and the definition of $\xi_n$, we have
\[
\CMV^t\delta_0(n)
=
\langle \delta_n, \CMV^t \delta_0 \rangle
=
\langle W\delta_n, W\CMV^t \delta_0 \rangle
=
\int_{\partial \D} \! z^t \overline{\xi_n(z)} \, d\mu (z)
\]
for any $t \in \Z$. Thus, we obtain that
\begin{align*}
P(n < M,T)
& =
(1-e^{-2/T}) \sum_{t=0}^{\infty} \sum_{n < M} e^{-\frac{2t}{T}} \int_{\partial \D} \! \int_{\partial \D} \! \left( \frac{z}{w} \right)^t
\overline{\xi_n(z)} \xi_n(w) \, d\mu(z) \, d\mu(w) \\
& =
(1-e^{-2/T}) \int_{\partial \D} \! \int_{\partial \D} \frac{1}{1-e^{-2/T}z/w} S_M(z,w) \, d\mu(z) \, d\mu(w). \\
\end{align*}
Applying Cauchy-Schwarz to the integral over $z$ and then to the integral over $w$, we obtain
\begin{align*}
P(n < M,T)	
& \leq
(1-e^{-2/T}) \int_{\partial \D} \! \int_{\partial \D} \! \left|\frac{1}{1-e^{-2/T}z/w} S_M(z,w) \right| \, d\mu(z)  \, d\mu(w) \\
&\leq
\int_{\partial \D} \sqrt{b(w,T) G_M(w)} \, d\mu(w) \\
& \leq
\sqrt{\int_{\partial \D} b (w, T)  \, d\mu(w) \int_{\partial \D} G_M(w)  \, d\mu(w)}.
\end{align*}
We now tackle the integrals of $b$ and $G_M$. First, since  $1 - e^{-2/T} \leq 1 - e^{-4/T}$, we have
\begin{align*}
\int_{\partial \D} \! b(w,T) \, d\mu(w)
& =
 (1-e^{-2/T})^2 \int_{\partial \D} \! \int_{\partial \D} \! \left |\frac{1}{1 - e^{-2/T}z/w}\right| ^2 d\mu(z) \, d\mu(w) \\
& \leq
(1-e^{-2/T}) \int_{\partial \D} \! \int_{\partial \D} \! P_{e^{-2/T}}(z/w) \, d\mu(z) \, d\mu(w) \\
& \leq
2J(1/T).
\end{align*}
To estimate the integral of $G_M$, notice that
\begin{align*}
\int_{\partial \D} \! G_M(w) \, d\mu(w)
& =
\int_{\partial \D} \! \int_{\partial \D} \! |S_M(z,w)|^2 \, d\mu(w) \, d\mu(z) \\
& =
\int_{\partial \D} \! \int_{\partial \D} \! \sum_{n < M} \sum_{m < M} \overline{\xi_n(z)}\xi_n(w) \xi_m(z) \overline{\xi_m(w)} \, d\mu(w) \, d\mu(z) \\
& =
\sum_{n,m< M} |\langle \xi_n,\xi_m\rangle|^2 \\
& \leq
M,
\end{align*}
where the final inequality follows from the definition of the generalized eigenfunctions and unitarity of $W$ (notice that we have not assumed $M$ is an integer). Thus, taking $M(T) = 1/(8J(1/T))$ as in the statement of the lemma, we get
\[
|P(n < M(T),T)|	
\leq
\sqrt{2J(1/T) \cdot M(T)} \\
=
\frac{1}{2}.
\]
\end{proof}

We will occasionally encounter the following integrals:
$$
I(\varepsilon)
\eqdef
\left(1 - e^{-2\varepsilon}\right) \int_0^{2\pi}\mathrm{Re}^2 \! \left( F \left(e^{i\theta+\varepsilon} \right) \right) \frac{d\theta}{2\pi}.
$$
The following lemma will allow us to relate $I$ to the return probabilities.

\begin{lemma}\label{l:jleqi}
For all $\varepsilon \geq 0$, $J(\varepsilon) \leq I(\varepsilon)$.
\end{lemma}

\begin{proof}
By Fubini's theorem, we have
\begin{align*}
I(\varepsilon)	
& =
\left(1 - e^{-2\varepsilon}\right) \int_0^{2\pi}\mathrm{Re}^2 \! \left( F \left(e^{i\theta+\varepsilon} \right) \right) \frac{d\theta}{2\pi} \\ 	
& =
\left(1 - e^{-2\varepsilon}\right)  \int_0^{2\pi} \int_{\partial \D} \int_{\partial \D} P_{e^{-\varepsilon}} \left( e^{i\theta} \overline{z}\right) P_{e^{-\varepsilon}} \left( e^{i\theta} \overline{w}\right) \, d\mu(z) \, d\mu(w) \, \frac{d\theta}{2\pi} \\
& =
\left(1 - e^{-2\varepsilon}\right) \int_{\partial \D} \int_{\partial \D} P_{e^{-2\varepsilon}} \left( \frac{z}{w} \right) \, d\mu(z) \,d\mu(w).
\end{align*}
The final equality is a straightforward calculation using the series definition of $P$ and the dominated convergence theorem. Using the definition of $P$ and dominated convergence once more, we obtain
\begin{align*}
I(\varepsilon)
& =
\left(1 - e^{-2\varepsilon}\right) \sum_{\ell \in \Z} e^{-2\varepsilon|\ell|} \int_{\partial \D} \int_{\partial \D} z^\ell w^{-\ell} \, d\mu(z) \,d\mu(w) \\
& =
\left(1 - e^{-2\varepsilon}\right) \sum_{\ell \in \Z} e^{-2\varepsilon |\ell|} | c_\ell|^2 \geq 1- e^{-2\varepsilon},
\end{align*}
where $c_\ell = \int_{\partial \D} z^\ell d\mu\left(z\right)$. Consequently,
\begin{align*}
J(\varepsilon)
& =
\frac{1-e^{-2\varepsilon}}{2} \int_{\partial \D} \int_{\partial \D} (1+P_{e^{-2\varepsilon}}(z/w)) \, d\mu(z) \, d\mu(w) \\
& =
\frac{1-e^{-2\varepsilon}}{2} + \frac{1}{2} I(\varepsilon) \\
&
\leq
I(\varepsilon),
\end{align*}
as desired.
\end{proof}

Next, we prove some bounds on the growth of the Gesztesy--Zinchenko matrices in the spirit of \cite[Lemma~2.3]{T2005}.

\begin{lemma}\label{l:znormbound}
Let $z = e^{i\theta + \varepsilon},$ with $\varepsilon \in (0, 1)$. Then, if $n\varepsilon \leq K$ for some constant $K > 0$, we have
\begin{align*}
\left\| Z(n,0;z) \right\|
& \leq
CL_N^{\frac{1 - \eta}{2\eta} \left(1 + 2\nu_N\right)} \\
\left\| Z(n,0;z) \right\|
& \leq
CL_{N+1}^{\frac{1 - \eta}{\eta} \nu_{N+1}}
\end{align*}
for all $n$ such that $L_N < n \leq L_{N+1}$, where $C = C_K$ is a constant which depends on $K$.
\end{lemma}

\begin{proof}
Suppose $L_N < n \leq L_{N+1}$. By definition of $Z(n,0;z)$, we have
\begin{equation} \label{eq:sparseZdecomp}
Z(n,0;z)
=
Z(n,L_N + 1;z) \prod_{j=N}^1 Y(L_j,z) Z(L_j,L_{j-1} + 1;z),
\end{equation}
where we adopt the convention $L_0 = -1$. Since $\varepsilon \in (0,1)$, we have $\|Y(L_j,z)\| \leq CL_j^{\frac{1-\eta}{2\eta}}$ for a constant $C$ by definition of $Y$ and $\alpha_{L_j}$. Additionally, we can see that $\|Y(n,z)\| \leq e^\varepsilon$ for all $n \notin \set{L_j : j \in \Z_+}$, which implies
\begin{align*}
\|Z(n,L_N + 1) \|
& \leq
e^{(n-L_N-1)\varepsilon},\\
\|Z(L_j,L_{j-1}+1)\|
& \leq
e^{(L_j-L_{j-1}-1)\varepsilon},
\text{ for all } 1 \leq j \leq N.
\end{align*}
Combining these estimates with \eqref{eq:sparseZdecomp}, we get
\[
\|Z(n,0;z)\|
\leq
C^N \prod_{j=1}^N L_j^{\frac{1-\eta}{2\eta}} e^{n\varepsilon}
\leq
C_K L_N^{\frac{1-\eta}{2\eta}(1 + 2\nu_N)},
\]
where we have applied the sparseness condition to get $L_1 \cdots L_{N-1} = L_N^{\nu_N}$ and $C^N \leq C' L_N^{\frac{1-\eta}{2\eta} \nu_N}$ for another constant $C' > 0$. The second bound is proved similarly using $L_1 \cdots L_N = L_{N+1}^{\nu_{N+1}}$.
\end{proof}

With the work above in hand, we can produce dynamical bounds similar to those in \cite[Theorem~2.4]{T2005}.

\begin{theorem}\label{t:pbounds}
Suppose $\frac{L_N}{4} \leq T \leq \frac{L_{N+1}}{4}$ for some $N \in \Z_+$. There exists a constant $C > 0$ independent of $N$ and $T$ such that
\begin{equation} \label{eq:lb:pbound1}
P(n \geq T,T)
\geq
CTL_N^{-\frac{1 - \eta}{\eta} (1 + 2\nu_N)} I(1/T)
\end{equation}
and
\begin{equation} \label{eq:lb:pbound3}
P(L_N/4 \leq n \leq L_N,T)
\geq
C L_N^{1 - 2\frac{1 - \eta}{\eta} \nu_N} I(1/T).
\end{equation}
\end{theorem}

\begin{proof} Suppose $\frac{L_N}{4} \leq T \leq  \frac{L_{N+1}}{4}$ and $T \leq n \leq 2T$, let $z = e^{i\theta + \frac{1}{T}}$, and define $R(z) \eqdef (\CMV - z)^{-1}$. Next, put $u(n,z) \eqdef R(z)\delta_0(n)$ and $v \eqdef \mathcal{M}u$. Notice that one immediately has
\begin{equation} \label{uvshift}
|u(2k-1,z)|^2 + |u(2k,z)|^2
=
|v(2k-1,z)|^2 + |v(2k,z)|^2
\end{equation}
for all $k \in \Z_+$, since $\mathcal M$ is a direct sum of $2\times 2$ unitary blocks. Now, by Proposition~\ref{p:cara:resolve}, we have
\[
\begin{pmatrix}
u(n,z) \\
v(n,z) \\
\end{pmatrix}
=
\frac{1}{2z} Z(n,0;z) \begin{pmatrix}
F(z) + 1 \\
F(z) - 1 \\
\end{pmatrix},
\]
where $F$ denotes the Carath\'eodory function of $\CMV$. Because $n/T \leq 2$, the previous lemma yields $\|Z(n,0;z)^{-1}\| =  \|Z(n,0;z)\| \leq CL_N^{\frac{1- \eta}{2\eta}(1+2\nu_N)}$. Consequently, since $|z| \leq e$, we get
\begin{align*}
|u(n,z)|^2 + |v(n,z)|^2
&\geq
CL_N ^{-\frac{1 - \eta}{\eta}(1+2\nu_N)} (|F(z) + 1|^2 + |F(z) - 1|^2) \\
& \geq
CL_N^{-\frac{1 - \eta}{\eta} (1+2\nu_N)} \Re^2(F(z)).
\end{align*}
Moreover, by \eqref{eq:parsevalprob} and the definitions of $u$ and $P$, we have
\[
P(n \geq T,T)
\geq
\left(e^{\frac{2}{T}} - 1\right) \int_0^{2\pi} \sum_{T \leq n \leq 2T} \left|u\left(n,e^{i\theta + \frac{1}{T}}\right)\right|^2 \frac{d\theta}{2\pi},
\]
where we have used nonnegativity of the summands to remove terms with $n > 2T$. Combining these two estimates, we have
\begin{align*}
P(n \geq T,T)
& \geq
\left(e^{\frac{2}{T}} - 1\right) \int_0^{2\pi} \sum_{T \leq n \leq 2T} \left|u\left(n,e^{i\theta + \frac{1}{T}}\right)\right|^2 \frac{d\theta}{2\pi} \\
& \geq
C T\left(e^{\frac{2}{T}} - 1\right) L_N^{-\frac{1 - \eta}{\eta} (1+2\nu_N)} \int_0^{2\pi} \Re^2\!\left(F\left(e^{i\theta + \frac{1}{T}}\right)\right) \frac{d\theta}{2\pi} \\
& \geq
C T L_N^{-\frac{1 - \eta}{\eta} (1+2\nu_N)} I(1/T),
\end{align*}
where we have used the definition of $I$ and nonnegativity thereof in the third line. Thus we have \eqref{eq:lb:pbound1}. Notice also that we have used \eqref{uvshift} in the second line to replace $|u(n)|^2$ with $|u(n)|^2 + |v(n)|^2$ at the expense of adjusting the constant in front.

Analogously, \eqref{eq:lb:pbound3} can be obtained by summing over $n$ with $\frac{L_N}{4} \leq n \leq L_N$.  Notice that one must apply Lemma~\ref{l:znormbound} with $N$ replaced by $N-1$ in this case.
\end{proof}

We can translate the bounds on the outside probabilities from the previous theorem into lower bounds on the moments directly.

\begin{coro}\label{c:momentbound}
Let $p > 0$, and suppose that $N$ and $T$ satisfy $\frac{L_N}{4} \leq T \leq \frac{L_{N + 1}}{4}$. Then the following bound holds:
\[
\left\langle \left| X \right|^p\right\rangle (T)
\geq
C_p I(1/T)^{-p} + C_p \left(L_N^{p + 1 - 2\frac{1 - \eta}{\eta} \nu_N} + T^{p + 1}L_N^{-\frac{1 - \eta}{\eta} (1 + 2\nu_N)}\right) I(1/T).
\]
\end{coro}

\begin{proof}
This is analogous to the proof of \cite[Corollary~2.7]{T2005}; we explain the details for the convenience of the reader. Observe that for any $M$ and $T$, we have
\[
\langle|X|^p\rangle(T)
\geq
M^p \sum_{n \geq M} \widetilde{a}(n,T)
=
M^p P(n \geq M,T).
\]
Let $M(T) = \frac{1}{8J(1/T)}$. By Lemma~\ref{l:firstbound}, $P(n \geq M(T),T) \geq \frac{1}{2}$, so, for each $T$, we have
\[
\langle|X|^p\rangle(T)
\geq
(M(T))^p P(n \geq M(T),T)
\geq
C_p J(1/T)^{-p}
\geq
C_p I(1/T)^{-p},
\]
where the last inequality comes from Lemma~\ref{l:jleqi}. Applying similar reasoning and \eqref{eq:lb:pbound1} from Theorem~\ref{t:pbounds}, we obtain
\[
\langle|X|^p\rangle(T)
\geq
T^p P(n \geq T,T)
\geq
CT^{p + 1}L_N^{-\frac{1 - \eta}{\eta} (1 + 2\nu_N)} I(1/T).
\]
Similarly, using \eqref{eq:lb:pbound3} from Theorem~\ref{t:pbounds}, we have
\[
\langle|X|^p\rangle(T)
\geq
C_p L_N^p P(n \geq L_N/4,T)
\geq
C_p L_N^{p + 1 - 2\frac{1 - \eta}{\eta} \nu_N} I(1/T).
\]
Therefore, by averaging the three lower bounds for $\langle|X|^p\rangle(T)$ and adjusting the constants, we obtain the desired result.
\end{proof}

These lower bounds on the moments suffice to prove our desired lower bounds on $\widetilde \beta^\pm$.

\begin{theorem}\label{t:betalowerbounds}
Let $p > 0$ and $\frac{L_N}{4} \leq T \leq \frac{L_{N + 1}}{4}$. The following estimate holds uniformly in $T$:
\begin{equation} \label{eq:momentlb}
\langle|X|^p\rangle(T)
\geq
C_p L_N^{-2\frac{1 - \eta}{\eta} \nu_N} \left(L_N^p + T^p L_N^{-\frac{p}{p + 1} \frac{1 - \eta}{\eta}}\right).
\end{equation}
In particular,
\[
\widetilde{\beta}^-(p) \geq \frac{p + 1}{p + 1/\eta}
\text{ and }
\widetilde{\beta}^+(p) \geq 1.
\]
\end{theorem}

\begin{proof}
We proceed as in \cite[Theorem~2.8]{T2005}. By Corollary~\ref{c:momentbound},
\[
\langle|X|^p\rangle(T)
\geq
C_p \! \left(y_T^{-p} + L_N^{-2\frac{1 - \eta}{\eta} \nu_N} \left(L_N^{p + 1} + T^{p + 1} L_N^{-\frac{1 - \eta}{\eta}}\right) y_T \right),
\]
where $y_T \eqdef I(1/T)$. If $f(y) = y^{-p} + Ky$, where $K$ is a positive number, then one can check that $f(y) \geq C_p K^{\frac{p}{p + 1}}$ for all $y > 0$ by an easy calculus exercise. Thus
\begin{align*}
\langle|X|^p\rangle(T)
&\geq
C_p \left(L_N^{-2\frac{1 - \eta}{\eta} \nu_N} \left(L_N^{p + 1} + T^{p + 1} L_N^{-\frac{1 - \eta}{\eta}}\right)\right)^{\frac{p}{p + 1}} \\
&\geq
C_p L_N^{-2\frac{1 - \eta}{\eta} \nu_N} \left(L_N^p + T^p L_N^{-\frac{p}{p + 1} \frac{1 - \eta}{\eta}}\right),
\end{align*}
which proves \eqref{eq:momentlb}. Notice that we have applied concavity of the function $y \mapsto y^{\frac{p}{p+1}}$ and $\frac{p}{p+1} < 1$ to obtain the second inequality. Now, let $s = \frac{p(1 - \eta)}{(p + 1) \eta}$. If $\frac{L_N}{4} \leq T \leq L_N^{\frac{p + s}{p}}$, then
\[
\langle|X|^p\rangle(T)
\geq
C_p L_N^{-2\frac{1 - \eta}{\eta} \nu_N} L_N^p
\geq
C_p L_N^{-2\frac{1 - \eta}{\eta} \nu_N} \left(T^{\frac{p}{p + s}}\right)^p
\geq
C_p T^{-2\frac{1 - \eta}{\eta} \nu_N + \frac{p^2}{p + s}}.
\]
Similarly, if $L_N^{\frac{p + s}{p}} \leq T \leq \frac{L_{N + 1}}{4}$, then
\[
\langle|X|^p\rangle(T)
\geq
C_p L_N^{-2\frac{1 - \eta}{\eta} \nu_N} T^p L_N^{-\frac{p}{p + 1} \frac{1 - \eta}{\eta}}
\geq
C_p T^{-2\frac{1 - \eta}{\eta} \nu_N + \frac{p^2}{p + s}},
\]
since $p - \frac{p^2}{(p + s)(p + 1)} \frac{1 - \eta}{\eta} = \frac{p^2}{p + s}$. Thus, for all $T$ with $\frac{L_N}{4} \leq T \leq \frac{L_{N + 1}}{4}$, we have
\[
\langle|X|^p\rangle(T)
\geq
C_p T^{-2\frac{1 - \eta}{\eta} \nu_N + \frac{p^2}{p + s}}.
\]
Therefore, since $\nu_N \to 0$ as $N \to \infty$, we have
\[
\widetilde{\beta}^-(p)
=
\liminf_{T \to \infty} \frac{\log \langle|X|^p\rangle(T)}{p \log T}
\geq \frac{p}{p + s}
=
\frac{p + 1}{p + 1/\eta},
\]
as desired. On the other hand, applying \eqref{eq:momentlb} to the sequence of time scales $T_N = L_N$, we get
\[
\widetilde{\beta}^+(p)
\geq
\limsup_{N \to \infty} \frac{\log \langle|X|^p \rangle(L_N)}{p \log L_N}
\geq
\limsup_{N \to \infty} \frac{\log\Big(C_p L_N^{-2\frac{1 - \eta}{\eta}\nu_N} L_N^p \Big)}{p \log L_N}
=
1,
\]
concluding the proof.
\end{proof}

\section{Proofs of Upper Bounds}

In this section, we prove upper bounds that complement the lower bounds of the previous section. Our main result is an adaptation of \cite[Theorem~3.4]{T2005} to the present setting.

\begin{theorem} \label{t:ubs}
Suppose that $T$ and $N$ satisfy $\frac{L_N}{4} \leq T \leq L_N^{\frac{1}{\eta}}$. Then for any $p \geq 0$,
\begin{equation} \label{eq:tubs:outmomentbound}
\sum_{n \geq 2L_N} n^p \widetilde a(n,T)
\leq
C_p T^{p + 1} L_N^{-\frac{1}{\eta}}.
\end{equation}
Additionally,
\begin{equation} \label{eq:tubs:poutbound}
P\left(n \geq 2L_N,T\right) \leq CTL_N^{-\frac{1}{\eta}}
\end{equation}
and
\begin{equation} \label{eq:tubs:momentbound}
\langle\abs{X}^p\rangle(T)
\leq
C_p L_N^p + C_p T^{p + 1} L_N^{-\frac{1}{\eta}}.
\end{equation}
\end{theorem}

We will break the proof into several smaller steps. We begin by noting that the contribution of sites beyond $n = T^2$ is negligible. In fact, this argument applies to all CMV matrices, not just those with sparse coefficients.

\begin{lemma}\label{l:tsquared}
Let $\CMV$ be any CMV matrix. For any $p \geq 0$, there is a constant $C_p > 0$ such that
\[
\sum_{n \geq T^2} n^p \widetilde a(n,T)
\leq
C_p e^{- \frac{T}{2}}
\]
for all $T > 0$.
\end{lemma}

\begin{proof}
First, note that $\CMV^t\delta_0(n)$ is 0 whenever $n > 2t$ because $\CMV$ is pentadiagonal. Thus, we have
\begin{align*}
\sum_{n \geq T^2} n^p \widetilde a(n,T)
& =
\sum_{n \geq T^2} n^p(1-e^{-\frac{2}{T}}) \sum_{t \geq \frac{n}{2}} e^{-\frac{2t}{T}}|\CMV^t\delta_0(n)|^2 \\
& =
\sum_{t \geq \frac{T^2}{2}}\sum_{T^2 \leq n \leq 2t} n^p(1-e^{-\frac{2}{T}})e^{-\frac{2t}{T}}|\CMV^t\delta_0(n)|^2 \\
& \leq
C_p \sum_{t \geq \frac{T^2}{2}} t^p e^{-\frac{2t}{T}},
\end{align*}
where we have used unitarity of $\CMV$ and $n \leq 2t$ in the final inequality. Since $u^p e^{-u} \to 0$ as $u \to \infty$, we have $t^p \leq C_p T^p e^{t/T}$ for all $t, T > 0$. Consequently,
\begin{align*}
\sum_{n \geq T^2} n^p \widetilde a(n,T)
& \leq
C_p \sum_{t \geq \frac{T^2}{2}} T^p e^{-t/T} \\
& \leq
C_p T^p \frac{e^{-T/2}}{1-e^{-1/T}} \\
& \leq
C_p e^{-\frac{T}{2}},
\end{align*}
where we have used $u^{p+1} e^{-u/2} \to 0$ as $u \to \infty$ in the final line.
\end{proof}
Let $\CMV$ be our sparse CMV matrix. For each $N \in \Z_+$, let $\CMV_N$ be the truncated CMV matrix, with coefficients $(\alpha_{N,n})_{n=0}^\infty$ defined by
$$
\alpha_{N,n}
=
\begin{cases}
\alpha(n) & n \leq L_N, \\
0 & \text{otherwise.}
\end{cases}
$$
Henceforth, denote $R(z) = (\CMV - z)^{-1}$ and $R_N(z) = (\CMV_N - z)^{-1}$ for each $N \in \Z_+$ and $z \notin \partial \D$. Our approach to Theorem~\ref{t:ubs} is as follows: in light of the Parseval formula \eqref{eq:parsevalprob}, we want to have good upper bounds on matrix elements of $R(z)$. We will accomplish this by proving effective upper bounds for matrix elements of $R_N(z)$ and $R(z) - R_N(z)$ and then applying the elementary fact that $\abs{a + b}^2 \leq 2\abs{a}^2 + 2\abs{b}^2$ for all $a,b \in \C$. We begin the estimation of the matrix elements of $R_N(z)$ by explicitly computing (some of) them.

\begin{lemma}\label{l:gh}
Fix $N \in \Z_+$ and $z = e^{i\theta + \varepsilon}$ with $\varepsilon > 0$. Let $g = R_N(z) \delta_0$ and $h = \mathcal M_N g$. For all $k \geq L_N+1$, we have
\begin{align}
\label{eq:grec}
g(k)
& =
\begin{cases}
z^{-\frac{k -L_N}{2}} h(L_N + 1) & \text{$k$ odd and $L_N$ odd,} \\
z^{-\frac{k - L_N - 1}{2}} g(L_N + 1) & \text{$k$ odd and $L_N$ even,} \\
0 & \text{$k$ even,}
\end{cases} \\
\label{eq:hrec}
h(k)
& =
\begin{cases}
0 & \text{$k$ odd,} \\
z^{-\frac{k - L_N - 1 }{2}} h(L_N + 1) & \text{$k$ even and $L_N$ odd,} \\
z^{-\frac{k - L_N - 2}{2}} g(L_N + 1) & \text{$k$ even and $L_N$ even.}
\end{cases}
\end{align}
\end{lemma}

\begin{proof}
From the proof of \eqref{eq:gzcara:resolvent} from Lemma~\ref{p:cara:resolve}, we have
$$
\begin{pmatrix}
g(k) \\
h(k)
\end{pmatrix}
=
Z_N(k,L_N + 1;z)
\begin{pmatrix}
g(L_N + 1) \\
h(L_N + 1)
\end{pmatrix}
$$
for $k \geq L_N + 1$, where $Z_N$ is the Gesztesy-Zinchenko transfer matrix for $\CMV_N$. Suppose first that $L_N$ is even. Because $\alpha_{N,j} = 0$ for all $j \geq L_N + 1$, we may explicitly compute $Z_N(k,L_N + 1; z)$, and we obtain
\begin{align*}
g(k)
& =
\begin{cases}
z^{-\frac{k - L_N - 1}{2}} g(L_N + 1) & \text{$k$ odd,} \\
z^{\frac{k - L_N - 2}{2}} h(L_N + 1) & \text{$k$ even,}
\end{cases} \\
h(k)
& =
\begin{cases}
z^{\frac{k - L_N - 1}{2}} h(L_N + 1) & \text{$k$ odd,} \\
z^{-\frac{k - L_N - 2}{2}} g(L_N + 1) & \text{$k$ even.}
\end{cases}
\end{align*}
for all $k \geq L_N + 1$. Because $\abs{z} = e^{\varepsilon} > 1$ and $g,h \in \ell^2(\Z_0)$, it follows that $h(L_N + 1) = 0$, which proves \eqref{eq:grec} and \eqref{eq:hrec} when $L_N$ is even. The arguments when $L_N$ is odd are identical.
\end{proof}

In light of the previous lemma, we can prove effective bounds on the ``tails'' of $g$ and $h$ by proving estimates on $g(L_N+1)$ and $h(L_N+1)$, which is what we accomplish in the following lemma.

\begin{lemma} \label{l:ghLNbound}
Fix $N \in \Z_+$, let $z = e^{i\theta + \varepsilon}$ with $0 < \varepsilon \leq \frac{4}{L_N}$, and put $g = R_N(z) \delta_0$ and $h = \mathcal{M}_N g$. Then
\begin{equation} \label{eq:ghLNbound}
\abs{g(L_N + 1)}^2 + \abs{h(L_N + 1)}^2
\leq
 -\frac{C}{1 + \left(e^{2\varepsilon} - 1\right) L_N^\frac{1}{\eta}} \Re\!\left(F_N(e^{i\theta + \varepsilon})\right),
\end{equation}
where $F_N$ denotes the Carath\'eodory function of $\mathcal C_N$.

\end{lemma}

\begin{proof}
Suppose that $L_N$ is even. First, we observe
\[
-\frac{1}{e^{2\varepsilon}-1} \mathrm{Re}(F_N(e^{i\theta + \varepsilon}))
=
\|g\|^2
\geq
\sum_{k = 0}^\infty |g(L_N+2k+1)|^2
=
\frac{1}{1-e^{-2\varepsilon}} |g(L_N+1)|^2
\]
by Proposition~\ref{p:cara:resolve} and Lemma~\ref{l:gh}, so we obtain
\begin{equation} \label{eq:easyghLNbound}
|g(L_N+1)|^2
\leq
- \mathrm{Re}(F_N(e^{i\theta + \varepsilon})).
\end{equation}
Now, for $\frac{L_N}{2} < k < L_N$, we have
\[
\begin{pmatrix}
g(k) \\
h(k) \\
\end{pmatrix}
=
Z(k,L_N;z) \begin{pmatrix}
g(L_N) \\
h(L_N) \\
\end{pmatrix}.
\]
Since $\alpha_k = 0$ for $L_{N-1} < k < L_N$, we get
\[
\|Z(k,L_N;z)\|
\leq
e^{\varepsilon(L_N-k)},
\]
so
\[
|g(k)|^2 + |h(k)|^2
\geq
e^{-2\varepsilon(L_N-k)} \left(\abs{g(L_N)}^2 + \abs{h(L_N)}^2\right).
\]
Because $\varepsilon \leq \frac{4}{L_N}$ and $\frac{L_N}{2} < k < L_N$, this yields
\begin{align*}
CL_N \left(\left|g(L_N)\right|^2 + \left|h(L_N)\right|^2\right)
& \leq
\|g\|^2 + \|h\|^2 \\
& =
2\|g\|^2 \\
& =
-\frac{2}{e^{2\varepsilon}-1} \Re\!\left(F_N\left(e^{i\theta+\varepsilon}\right)\right),
\end{align*}
where the second line follows from unitarity of $\mathcal{M}_N$ and the final line follows from Proposition~\ref{p:cara:resolve}. Consequently,
\begin{equation} \label{eq:gLN1stbound}
\left|g\left(L_N\right)\right|^2
\leq
-\frac{C}{\left(e^{2\varepsilon} - 1\right) L_N} \Re\!\left(F_N\left(e^{i\theta + \varepsilon}\right)\right).
\end{equation}

For simplicity, denote $\alpha = \alpha_{L_N} = \sqrt{1 - L_N^{-\frac{1 - \eta}{\eta}}}$ and $\rho = \rho_{L_N} = L_N^{-\frac{1 - \eta}{2\eta}}$. Since $L_N$ is even, we have
\begin{equation} \label{eq:ghLNghLNp1}
\begin{pmatrix}
g(L_N + 1) \\
h(L_N + 1)
\end{pmatrix}
=
\frac{1}{\rho}
\begin{pmatrix}
-\alpha & z^{-1} \\
z & -\alpha
\end{pmatrix}
\begin{pmatrix}
g(L_N) \\
h(L_N)
\end{pmatrix}.
\end{equation}
Because $h(L_N + 1) = 0$, we have $h(L_N) = \frac{z}{\alpha} g(L_N)$, so the first component of \eqref{eq:ghLNghLNp1} gives us
\[
g(L_N + 1)
=
\frac{1}{\rho} \left(-\alpha + \alpha^{-1}\right) g(L_N)
=
\frac{\rho}{\alpha} g(L_N).
\]
Consequently,
\[
\abs{g(L_N + 1)}^2
=
\frac{\rho^2}{\alpha^2} \abs{g(L_N)}^2
=
\frac{L_N^{-\frac{1 - \eta}{\eta}}}{1 - L_N^{-\frac{1 - \eta}{\eta}}} \abs{g(L_N)}^2
\leq
2L_N^{-\frac{1 - \eta}{\eta}} \abs{g(L_N)}^2,
\]
where the last step requires taking $N$ sufficiently large. Combining this with \eqref{eq:gLN1stbound}, we have
\begin{equation} \label{eq:gLNbound1}
\abs{g(L_N + 1)}^2
\leq
-\frac{C}{\left(e^{2\varepsilon} - 1\right) L_N^{\frac{1}{\eta}}} \Re\left(F_N\left(e^{i\theta + \varepsilon}\right)\right).
\end{equation}
Thus, using \eqref{eq:easyghLNbound}, \eqref{eq:gLNbound1}, and $h(L_N + 1) = 0$, we obtain \eqref{eq:ghLNbound} in this case. The argument when $L_N$ is odd is identical, except $g(L_N + 1) = 0$ and we need to bound $h(L_N + 1)$.
\end{proof}

We now turn towards the estimation of matrix elements of $R(z) - R_N(z)$. In light of standard resolvent identities, we first compute $\CMV - \CMV_N$.

\begin{lemma}
\label{cmv-phi}
Let $\phi \in \Hi$ be given. For each $k \in \Z_+$, define $v_k = v_k(\phi) \in \Hi$ as follows. If $L_k$ is even, let
\[
v_k
\eqdef
(\alpha_{L_k}\phi_{L_k-1} + (\rho_{L_k}-1)\phi_{L_k+2})\delta_{L_k}
+
((\rho_{L_k}-1)\phi_{L_k-1} - (\alpha_{L_k}\phi_{L_k+2})\delta_{L_k+1},
\]
and, if $L_k$ is odd, let
\[
v_k
\eqdef
(\alpha_{L_k}\phi_{L_k} + (\rho_{L_k}-1)\phi_{L_k+1})\delta_{L_k-1}
+
((\rho_{L_k}-1)\phi_{L_k} - (\alpha_{L_k}\phi_{L_k+1})\delta_{L_k+2}.
\]
Then
\begin{equation} \label{eq:truncdiff}
(\CMV - \CMV_N)\phi
=
\sum_{k=N+1}^\infty v_k
\end{equation}
for all $N$ sufficiently large.
\end{lemma}

\begin{proof}
Let us suppose that $N_0$ is large enough that $L_{k+1} \geq L_k + 2$ for all $k \geq N_0$. Then, for all $N \geq N_0$, \eqref{eq:truncdiff} follows from straightforward calculations using the definitions of $\CMV$ and $\CMV_N$.
\end{proof}

The explicit form of $\CMV - \CMV_N$ together with the sparseness condition on the $L$'s enables us to prove quite strong estimates on $\|R(z) \delta_0 - R_N(z) \delta_0 \|$.

\begin{lemma} \label{l:rdiffbound}
Fix $N \in \Z_{+}$ sufficiently large, and let $z = e^{i\theta+\varepsilon}$ with $L_N^{-\frac{1}{\eta}} < \varepsilon < \frac{4}{L_N}$. Then
\[
\| R(z) \delta_0 - R_N(z) \delta_0 \|^2
\leq
Ce^{-\varepsilon^{-\eta}}.
\]
\end{lemma}

\begin{proof}
As before, we denote $g = R_N(z) \delta_0$. First, observe that $\|R(z)\| \leq (e^\varepsilon-1)^{-1} \leq \varepsilon^{-1}$, so
\begin{equation}
\label{4-7ineq}
\| R(z) \delta_0 - R_N(z) \delta_0 \|^2
\leq
\varepsilon^{-2} \|(\CMV - \CMV_N)g\|^2
\end{equation}
by a standard resolvent identity. Using Lemmas~\ref{l:gh} and \ref{cmv-phi}, we have
\begin{align*}
\norm{(\CMV - \CMV_N)g}^2
&\leq
C \sum_{k = N + 1}^{\infty} e^{-2\varepsilon (L_k - L_N)} \left(\alpha_{L_k}^2 + (\rho_{L_k} - 1)^2\right) \frac{1}{\left(e^{\varepsilon} - 1\right)^2} \\
& \leq
C\varepsilon^{-2}  \sum_{k = N + 1}^{\infty} e^{-2\varepsilon (L_k - L_N)}.
\end{align*}
Combining this with (\ref{4-7ineq}), we have
\[
\norm{R(z)\delta_0 - R_N(z)\delta_0}^2
\leq
C\varepsilon^{-4} \sum_{k = N + 1}^{\infty} e^{-2\varepsilon (L_k - L_N)}.
\]
Now, use the sparseness condition to ensure that $N$ is large enough that $L_{N + 1} \geq 2L_N$, so that
\[
2\varepsilon \left(L_k - L_N\right)
\geq
\varepsilon L_k
\geq
L_N^{-\frac{1}{\eta}} L_k
\]
for all $k \geq N + 1$. Using sparseness once more to get $L_{N + 1} \geq L_N^{\frac{1}{\eta} + 1}$, we have
\[
\sum_{k = N + 1}^{\infty} e^{-2\varepsilon (L_k - L_N)} \leq
\sum_{k = N + 1}^{\infty} e^{- L_N^{-\frac{1}{\eta}} L_k}
\leq
\sum_{m = L_{N + 1}}^{\infty} e^{- L_N^{-\frac{1}{\eta}} m}
\leq
 Ce^{-L_N}.
\]
To obtain the final inequality, we have summed the geometric series and used $L_{N+1} \geq L_N^{\frac{1}{\eta} + 1}$. Since $\varepsilon \geq L_N^{-\frac{1}{\eta}}$ and $\varepsilon^{-4}e^{-\varepsilon^{-\eta}} \to 0$ as $\varepsilon \downarrow 0$, we may enlarge $C$ to obtain the statement of the lemma.
\end{proof}

We now have all of the estimates which we will need to prove Theorem~\ref{t:ubs}.

\begin{proof}[Proof of Theorem~\ref{t:ubs}]
By \eqref{eq:parsevalprob}, we have
\[
\widetilde a(n,T)
=
\left(e^{\frac{2}{T}} - 1\right) \int_0^{2\pi} \abs{R\left(e^{i\theta + \frac{1}{T}}\right)\delta_0(n)}^2 \frac{d\theta}{2\pi}.
\]
Since $|a+b|^2 \leq 2|a|^2 + 2|b|^2$ for all complex numbers $a$ and $b$, we may bound the integral on the right hand side by bounding the same integral with $R$ replaced by $R_N$ and $R-R_N$ and combining the resulting inequalities. First, apply Lemmas~\ref{l:gh} and \ref{l:ghLNbound} and $(e^{2/T}-1)^{-1} \leq T/2$ to get
\[
\sum_{2L_N \leq n \leq T^2} n^p \abs{g(n)}^2 \\
\leq
-CTL_N^{-\frac{1}{\eta}} \! \sum_{L_N \leq k \leq \frac{T^2 - 1}{2}} \! \left(2k + 1\right)^p e^{\frac{1}{T} (L_N - 2k)} \mathrm{Re}(F_N(e^{i\theta+1/T})),
\]
where $g = R_N(e^{i\theta+1/T}) \delta_0$, as before. Additionally, because $\frac{L_N}{T} \leq 4$,
\begin{align*}
\sum_{L_N \leq k \leq \frac{T^2 - 1}{2}} \left(2k + 1\right)^p e^{\frac{1}{T} \left(L_N - 2k\right)}
& \leq
C_p \sum_{L_N \leq k \leq \frac{T^2 - 1}{2}} k^p e^{-\frac{2k}{T}} \\
& \leq
C_p T^p \sum_{L_N \leq k \leq \frac{T^2 - 1}{2}}  e^{-\frac{k}{T}} \\
& \leq
C_p T^{p + 1}.
\end{align*}
Thus, by Lemma~\ref{l:cara:int}, we have
\begin{equation} \label{eq:ubthm:RNbound}
\int_0^{2\pi} \sum_{2L_N \leq n \leq T^2} n^p \abs{R_N\left(e^{i\theta + \frac{1}{T}}\right)\delta_0(n)}^2 \frac{d\theta}{2\pi} \\
\leq
C_p T^{p + 2} L_N^{-\frac{1}{\eta}}.
\end{equation}

Now, we turn towards bounding the relevant integral with $R$ replaced by $R - R_N$. Denote $\varphi(n,z) = R(z)\delta_0(n) - R_N(z) \delta_0(n)$ for $z \notin \partial \D$ and $n \in \Z_0$. By Lemma~\ref{l:rdiffbound}, we have
\begin{align*}
\sum_{2L_N \leq n \leq T^2} n^p \abs{\varphi(n,e^{i\theta + 1/T})}^2
& \leq
T^{2p}  \norm{\varphi(\cdot,e^{i\theta+1/T})}^2
\\
& \leq
C T^{2p} e^{-T^{\eta}}
\end{align*}
for each $\theta \in [0,2\pi]$. Combining this estimate with \eqref{eq:ubthm:RNbound}, and using the Parseval formula from \eqref{eq:parsevalprob}, we see that
\[
\sum_{2L_N \leq n \leq T^2} n^p \widetilde a(n,T)
\leq
CT^{2p} e^{-T^{\eta}} + C_p T^{p + 1} L_N^{-\frac{1}{\eta}}.
\]
for all sufficiently large $N$. Because $L_N^{-1/\eta} \geq (4T)^{-1/\eta}$, the second term dominates in the limit $T \to \infty$, so we may enlarge $C_p$ to ensure that
\[
\sum_{2L_N \leq n \leq T^2} n^p \widetilde a(n,T)
\leq
C_p T^{p + 1} L_N^{-\frac{1}{\eta}}.
\]
Using Lemma~\ref{l:tsquared}, to bound the sum over $n \geq T^2$, \eqref{eq:tubs:outmomentbound} from the theorem follows. By letting $p = 0$, we obtain the bound on the outside probabilities from \eqref{eq:tubs:poutbound}
Using this, we have
\begin{align*}
\langle\abs{X}^p\rangle(T)
& =
\sum_{n} (n^p + 1) \widetilde a(n,T) \\
&=
\sum_{n < 2L_N} (n^p + 1) \widetilde a(n,T)
+ \sum_{n \geq 2L_N} (n^p + 1) \widetilde a(n,T) \\
&\leq
C_p L_N^p + C_p T^{p + 1} L_N^{-\frac{1}{\eta}},
\end{align*}
which proves \eqref{eq:tubs:momentbound}.

\end{proof}

\begin{coro}\label{c:betaupperbounds}
For all $p > 0$,
\[
\widetilde{\beta}^-(p) \leq \frac{p + 1}{p + 1/\eta}
\quad
\text{ and }\quad
\widetilde{\beta}^+(p) \leq 1.
\]
\end{coro}

\begin{proof}
From Theorem~\ref{t:ubs}, we have the upper bound
\[
\langle\abs{X}^p\rangle(T)
\leq
C_p L_N^p + C_p T^{p + 1} L_N^{-\frac{1}{\eta}}
\]
whenever $L_N/4 \leq T \leq L_N^{1/\eta}$. Let $s = \frac{p + 1/\eta}{p + 1}$ and $T_N = L_N^s$. Since $s(p+1) - \frac{1}{\eta} = p$, we get $\langle |X|^p \rangle(T_N) \leq C_p L_N^p$, and hence
\[
\widetilde{\beta}^-(p)
\leq
\liminf_{N \to \infty} \frac{\log \langle\abs{X}^p\rangle(T_N)}{p \log T_N}
\leq
\liminf_{N \to \infty} \frac{\log \left(C_p L_N^p \right)}{p \log (L_N^s)}
=
\frac{1}{s}
=
\frac{p + 1}{p + 1/\eta}.
\]
The upper bound on $\widetilde{\beta}^+(p)$ follows immediately from the observation that $\langle \delta_n, \CMV^t \delta_0 \rangle = 0$ whenever $n > 2t$.
\end{proof}

With these upper bounds on $\widetilde{\beta}^+(p)$ and $\widetilde{\beta}^-(p)$ in hand, we have our main result.

\begin{proof}[Proof of Theorem~\ref{t:betas}]
This is immediate from Theorem~\ref{t:betalowerbounds} and Corollary~\ref{c:betaupperbounds}.
\end{proof}

Now that we have our upper bounds on $\widetilde \beta^-(p)$, we may deduce the claimed result on the fractal dimension of the spectral measure $\mu$:

\begin{proof}[Proof of Theorem~\ref{t:specdim}]
By Theorem~\ref{t:betas} and \cite[Corollary~3.13]{DFV}, we have
$$
\mathrm{dim}_{\mathrm{H}}^+(\mu)
\leq
\widetilde\beta^-(p)
=
\frac{p+1}{p+1/\eta}
$$
for all $p > 0$. Sending $p \downarrow 0$, we obtain $\mathrm{dim}_{\mathrm{H}}^+(\mu) \leq \eta$. When combined with Theorem~\ref{t:subord}, this yields the desired result.
\end{proof}

\begin{appendix}

\section{Spectral Continuity via Subordinacy Theory} \label{sec:speccont}

In this appendix, we briefly descbribe how one may verify Theorem~\ref{t:subord}. First, one needs estimates on the norms of the so-called Szeg\H{o} transfer matrices. If $(\alpha_n)_{n=0}^\infty$ is a sequence with $\alpha_n \in \D$ for all $n$, the corresponding Szeg\H{o} transfer matrices are defined by
$$
T(n,m;z)
\eqdef
\begin{cases}
S(\alpha_{n-1},z) \times \cdots \times S(\alpha_m,z) & n > m \\
I & n = m \\
T(m,n;z)^{-1} & n < m
\end{cases}
$$
where
$$
S(\alpha,z)
\eqdef
\frac{1}{\rho}
\begin{pmatrix}
z & -\overline{\alpha} \\
-\alpha z & 1
\end{pmatrix},
\quad
\alpha \in \D,
\quad
\rho = \sqrt{1-|\alpha|^2}.
$$
Now, let $L_1 < L_2 < \cdots$ denote a sequence of positive integers satisfying \eqref{eq:sparsedef}, and let $\CMV$ denote the corresponding CMV matrix defined by \eqref{eq:sparseCMVdef}. If $L_N < n \leq L_{N+1}$, then the definition of $T$ immediately implies that
\[
T(n,0;z)
=
T(n,L_N + 1) \times \prod_{j=N}^1 Y(L_j,z)  T(L_j,L_{j-1}+1;z),
\]
where we adopted the convention $L_0 = -1$. Using this representation, it is straightforward to prove upper and lower bounds on $T(n,0;z)$ in the same vein as \cite{JL99}.

\begin{prop} \label{p:szbounds}
Let $\gamma \eqdef \frac{1-\eta}{2\eta}$. For all $\varepsilon > 0$, there exists $N_0$ such that
\begin{equation} \label{eq:szegobounds}
L_N^{\gamma - \varepsilon}
\leq
\|T(n,0;z)\|
\leq
L_N^{\gamma + \varepsilon}
\end{equation}
for all $z \in \partial\D$ whenever $L_N < n \leq L_{N+1}$ and $N \geq N_0$.
\end{prop}

\begin{proof}
Suppose $L_N < n \leq L_{N+1}$ and $z \in \partial \D$. Naturally, $\|S(\alpha_k,z)\| = 1$ whenever $k \notin\set{L_j : j \in \Z_+}$. On the other hand, one has $\|S(\alpha_{L_j},z)\| \leq CL_j^\gamma$ for some constant $C > 0$. Thus,
$$
\|T(n,0;z)\|
\leq
C^N \prod_{j=1}^N L_j^\gamma
=
C^N L_N^{\gamma(1+\nu_N)},
$$
where the second equality uses the definition of $\nu_N$. Taking $N$ sufficiently large and using the sparseness condition, we get $C^N \leq L_N^{\varepsilon/2}$ and $\nu_N \gamma \leq \varepsilon/2$, which proves the upper bound from \eqref{eq:szegobounds}. To prove the desired lower bound, notice that we have the following for all sufficiently large $N$:
\begin{align*}
\|T(n,0;z)\|
& \geq
\|T(n,L_N;z)\| \|T(L_N,0;z)\|^{-1} \\
& \geq
C^{-1} L_N^\gamma L_{N-1}^{-\gamma-1} \\
& \geq
C^{-1} L_N^{\gamma-\nu_N(\gamma+1)},
\end{align*}
where we have applied the upper bound from \eqref{eq:szegobounds} with $\varepsilon = 1$ and $n=L_{N}$ (notice that one must replace $N$ by $N-1$ to use the bound as stated). As before, we have $C^{-1} \geq L_N^{-\varepsilon/2}$ and $\nu_N(\gamma+1) \leq \varepsilon/2$, provided $N$ is sufficiently large.
\end{proof}

With these bounds in hand, one can easily prove suitable upper and lower bounds on the growth of the first and second kind orthogonal polynomials, given by
$$
\begin{pmatrix} \varphi_n(z) \\ \varphi_n^*(z) \end{pmatrix}
=
T(n,0;z) \begin{pmatrix} 1 \\ 1 \end{pmatrix},
\quad
\begin{pmatrix} \psi_n(z) \\ \psi_n^*(z) \end{pmatrix}
=
T(n,0;z) \begin{pmatrix} 1 \\ -1 \end{pmatrix}.
$$
Let us briefly recall some notation from subordinacy theory. Given a sequence $a:\Z_0 \to \C$ and $m \geq 0$, we denote $\underline m = \lfloor m \rfloor$, and
$$
\|a\|_m^2
=
\sum_{j=0}^{\underline m} |a(j)|^2
+
\set{m} \left|a\left(\underline m + 1 \right) \right|^2,
$$
where $\set{m} = m - \underline m$ denotes the fractional part of $m$. Thus, $\|a\|_m$ is a local $\ell^2$ norm for $m \in \Z_0$ and one linearly interpolates $\|a\|_m^2$ between consecutive integers.

We can use the bounds from Proposition~\ref{p:szbounds} to produce bounds on the local $\ell^2$ norms of the sequences $\varphi(z) = (\varphi_n(z))_{n=0}^\infty$ and $\psi(z) = (\psi_n(z))_{n=0}^\infty$.

\begin{prop} \label{p:polybounds}
Let $\beta = \eta/(2-\eta)$. For all $\delta > 0$ and all $z \in \partial \D$,
\begin{equation}
\liminf_{m \to \infty} \frac{\|\varphi(z)\|_m^2}{\|\psi(z)\|_m^{2(\beta-\delta)}}
>
0.
\end{equation}
\end{prop}

\begin{proof}
Simply follow the proof of \cite[(5.9)]{JL99}, using Proposition~\ref{p:szbounds} instead of \cite[(5.7)]{JL99}. To rerun their arguments, it suffices to note that
$$
|\varphi_n^*(z)|
=
|\varphi_n(z)|,
\quad
|\psi_n^*(z)|
=
|\psi_n(z)|
$$
for every $n \geq 0$ and every $z \in \partial \D$, which can readily be seen from the identities
$$
\varphi_n^*(z)
=
z^n \overline{\varphi_n(1/\overline{z})},
\qquad
\psi_n^*(z)
=
z^n \overline{\psi_n(1/\overline{z})},
\qquad
n \geq 0, \, z \in \C.
$$
\end{proof}

\begin{proof}[Proof of Theorem~1.2]
By Proposition~\ref{p:polybounds} and \cite[Theorem~10.8.5]{S2},
$$
\limsup_{\varepsilon \downarrow 0} \frac{\mu\set{z_0e^{i\theta} : \theta \in (-\varepsilon,\varepsilon)}}{(2\varepsilon)^{\eta-\delta}}
<
\infty
$$
for all $z_0 \in \partial \D$ and all $\delta > 0$. The conclusion of the theorem follows from \cite[Theorem~10.8.7]{S2}.
\end{proof}

\end{appendix}

\end{document}